\documentclass{svproc}
\usepackage{graphicx}
\usepackage{amsmath,amssymb,amsfonts}
\usepackage{booktabs}
\usepackage{tikz}
\usepackage{pgfplots}
\pgfplotsset{compat=1.16}
\usetikzlibrary{positioning,arrows.meta}
\usepackage[hyphens]{url}

\begin{document}
\mainmatter

\title{Series and Parallel Dependency Links:\\ Vitality for Interdependent Networks of Networks}
\titlerunning{Series and Parallel Dependency Links}

\author{Cristian Curaba \and Stefano Grimaz}
\authorrunning{C. Curaba and S. Grimaz}

\institute{University of Udine, Udine, Italy,\\
\email{cristian.curaba@uniud.it}}

\maketitle

\begin{abstract}
In a network of networks, a dependency link runs from a supplier to the node it
serves and binds the supplier's failure to the node's failure. Indices computed from the
graph topology leave aside the service each edge carries, but in real
interdependent systems that service decides how failure propagates: a pumping
station requires electricity and telemetry, and losing either stops it; a
district fed by either of two reservoirs survives the loss of one.
Building on this observation, we label each node of the network with the service it provides.
We then measure how far this qualitative input can change the node's vitality ranking: 
by fixing the graph topology we show that changing only the labels moves the Mean Operativity Loss from $0.023$ to $0.114$,
while a ranking computed from the unlabelled graph topology is blind to the
difference. The links give a node its suppliers; the labels group those suppliers
into alternatives within one service and joint requirements across services. 
These labels should belong to the description of
an interdependent system alongside the links themselves for a fidelity vitality measure\footnote{Code for reproducibility: https://github.com/cascade-platform-org/CASCADE.}.

\end{abstract}

\section{Introduction}
\label{sec:intro}

Research on networks of networks separates two kinds of
edge~\cite{parshani_critical_2011,bashan_percolation_2011}. A \emph{connectivity
link} lets nodes function together as a network; a \emph{dependency link} binds
the failure of one node to the failure of another. We focus on dependency networks.

Dependency links can be interpreted in two ways. Under the first, failure propagates
along every dependency link. Parshani et
al.~\cite{parshani_critical_2011} write that for dependency links ``when a node
fails, his direct neighbors also fail (with some probability)''. Bashan et
al.~\cite{bashan_percolation_2011} put the same rule in group form, with the
``failure of an entire dependency groups [sic] due to a failure of one
member''.
Buldyrev et al.~\cite{buldyrev_catastrophic_2010} pair each node with exactly one
node of the other network, so that ``if node $A_i$ stops functioning owing to
attack or failure, node $B_i$ stops functioning''. Under the second, redundancy
rescues a node: Shao et al.~\cite{shao_cascade_2011} assume that ``each node in
one network can function only if it has at least a single support link connecting
it to a functional node in the other network''.

Reliability theory has a name for each interpretation. The first is a \emph{series system}: the
whole fails when any part fails. The second is a \emph{parallel} or
$1$-out-of-$n$ system: the whole survives while any one part
survives~\cite{rausand_system_2008}. Each rule quoted above is stated over all
the nodes or all the dependency links of its model, so within any one of those
papers every dependency link is a series element, or every one of them is
parallel.

Real interdependent systems generally contain links of both kinds, and what a
link delivers decides which kind it is. A pumping station requires electricity
\emph{and} telemetry, and losing either stops it: two links, two services, a
series system. A district fed by either of two reservoirs survives the loss of
one: two links, one service, a parallel system. The two pictures are
indistinguishable by node count, degree, or any other property of the graph
topology alone. This yields the observation the paper is built on:

\begin{quote}
Two dependency links entering a node are \emph{alternatives} when they carry the
same service and are \emph{jointly required} when they carry different services.
\end{quote}

The observation asks one qualitative label in the input graph---the service each node
provides---and this paper measures how far that input changes the node's vitality ranking. Three
things follow, one per section.

\emph{One qualitative label per node defines the requirements of every node}
(Section~\ref{sec:model}). The graph should record the service each node provides,
and one assumption carries the rest: a link into $v$ from a node providing
service $s$ makes $s$ a requirement of $v$. Given the labels, what $v$ requires
is topological: it is read off the incoming links. That set decides whether two
links entering $v$ are alternatives or are both needed, so it belongs to the
network's description from then on, alongside the links and the labels.

\emph{Node importance as vitality on the operativity} (Section~\ref{sec:vitality}).
Remove a node and measure the operativity lost. The requirements stay put while
the node goes, so the neighbours that depended on it still need what it supplied
and fail if no alternative is available.

\emph{The grouping of the supplier nodes is what the labels provides}
(Section~\ref{sec:benchmarks}). A graph topology gives the number and the
arrangement of the links entering a node; the labels partition them into
alternatives and joint requirements. Holding one graph topology fixed and varying
only the service labels moves the Mean Operativity Loss by a factor of five while
every index computed from the graph topology necessarily remains the same.

\section{Related Work}
\label{sec:related}

\paragraph{Dependency links.}
The connectivity/dependency distinction and the series reading are due to
Parshani et al.\ and Bashan et al.~\cite{parshani_critical_2011,bashan_percolation_2011},
building on the coupled-network percolation model of Buldyrev et
al.~\cite{buldyrev_catastrophic_2010} and its network-of-networks
generalisation~\cite{gao_networks_2012}. Shao et al.\ supply the parallel
reading, in which a node survives on any one of several support
links~\cite{shao_cascade_2011}; redundancy in dependency links is therefore
established. Radicchi and Bianconi place the threshold between the
two, holding a node functional while it functions in at least two of its layers,
so that further layers raise the system's
robustness~\cite{radicchi_redundant_2017}; that threshold is one number applied
to every node. A systematic review of frameworks published between 2005 and 2017
reports that the field concentrates on models in which ``a node in one layer
interacts with exactly one node at another layer'', and names multiple
interactions as an open opportunity~\cite{bachmann_survey_2020}. We use the
multilayer vocabulary of Kivel\"a et al.~\cite{kivela_multilayer_2014}: each set of nodes
grouped by the service they provide forms one layer.

\paragraph{Per-link variation.}
Each model above fixes one reading and applies it to every dependency link it
contains. A second line of work already varies the reading link by link, so that
one dependency link can behave differently from its neighbours. A
\emph{reinforced} link keeps its consumer running even after its supplier has
failed, so the dependency that link carries is switched
off~\cite{zhang_percolation_2024,li_percolation_2022}. The difference we make is
of another kind. Reinforcement decides whether a single link imposes its
dependency at all; a service label decides how two links entering the same node
combine, as alternatives or as a pair both of which are needed.

\paragraph{Graded operability.}
Functional Dependency Network Analysis (FDNA) propagates a continuous operability
through a supplier--provider graph~\cite{garvey_introduction_2010,garvey_modelling_2014},
after the inoperability input--output line~\cite{haimes_leontief-based_2001}.
Guariniello and DeLaurentis give the aggregation over a node's suppliers as
$O_j = \min(\mathrm{SOD}_{O_j}, \mathrm{COD}_{O_j})$: the first term averages
$\alpha_{ij}O_i + (1-\alpha_{ij})\mathrm{SE}_j$ over the suppliers $i$,
blending each one's operability with the node's own by the strength
$\alpha_{ij}$ of the link; the second minimises $O_i + \beta_{ij}$ over
them~\cite{guariniello_dependency_2013}. FDNA reaches a range of behaviours by
tuning those two numbers on every link, so its propagation is customised
numerically while one aggregation serves the whole model. The second term caps
$O_j$ at $O_i + \beta_{ij}$ for each supplier separately, so a single failed
supplier holds the node down whatever the others do, and lifting that cap means
setting the criticality $\beta_{ij}$ to its maximum, where the link imposes
nothing.

\paragraph{And--Or structure.}
The rule of Section~\ref{sec:model} can be reconstructed from a known object. Bundle the providers of
one service into a single arrow into the node requiring it, so that a node with
two required services has two arrows entering it and an arrow carries failure
only when every provider in its bundle has failed. Arrows that run from a set of
nodes to one node are the $B$-hyperarcs of Gallo et al., who use them for
And--Or graphs~\cite{gallo_directed_1993}. Growing a set from the nodes failed
at the outset, and adding a node whenever some arrow into it has its whole
bundle inside, closes on exactly the nodes the model leaves failed: failure here
is hypergraph reachability.

\paragraph{Vitality, and where its invariant comes from.}
Koschützki et al.\ define the class we use~\cite{koschutzki_centrality_2005}.
Choose any function that scores a whole network---its \emph{invariant
function}---and a node's centrality is how far that score drops when the node is
removed. The measure of Section~\ref{sec:vitality} fits this definition.
Different choices of invariant give different indices: Latora and Marchiori choose
network efficiency and apply the result to
infrastructure~\cite{latora_vulnerability_2005}. Skibski characterises the class
axiomatically~\cite{skibski_vitality_2021}. Which invariant to
choose is then the open question, and Borgatti gives the reason it
matters~\cite{borgatti_centrality_2005}: every centrality makes an assumption
about how the thing being measured moves through the network, and an index whose
assumption is wrong for the system at hand ranks the wrong nodes. 
Klemm et al.\ and Liu et al.\ score a node by how its state propagates under a spreading dynamics, 
so the ranking they produce is a function of the process they run~\cite{klemm_measure_2012,liu_locating_2016}.
The invariant of Section~\ref{sec:vitality} comes from a model in the same way: $N$ counts the nodes Eq.~(\ref{eq:model}) leaves operational.
\section{The Model}
\label{sec:model}

Let $G=(V,E,\sigma)$ be a directed graph in which each node $v$ provides one
\emph{service} $\sigma(v)$; the nodes providing service $s$ form the \emph{layer}
for $s$. A link $u \to v$ reads ``$u$ provides to $v$''.

The services a node requires are read off the network,
\begin{equation}
S_{\mathrm{req}}(v) \;=\; \{\, \sigma(u) \;:\; u \to v \,\},
\label{eq:requires}
\end{equation}
so $v$ requires service $s$ because a link points to it from a node providing $s$.
Write
\begin{equation}
P_s(v) \;=\; \{\, u \in V \;:\; u \to v, \;\; \sigma(u) = s \,\}
\label{eq:suppliers}
\end{equation}
for the suppliers of $s$ to $v$, so the incoming links of $v$ partition into the
groups $\{P_s(v)\}_{s \in S_{\mathrm{req}}(v)}$, one per required service. That
partition is what the labels contribute.
Equation~(\ref{eq:requires}) is evaluated once, on the network as given, and
$S_{\mathrm{req}}(\cdot)$ belongs to the description of the dependency graph from then on: we write
$D = (V, E, \sigma, S_{\mathrm{req}})$ for that description and work with it in place of
$G$ below.

Failure propagates through $D$ as a Boolean state on the nodes. An
\emph{initial state} $f_0 : V \to \{0,1\}$ gives every node $v$ a value $f_0(v)$,
equal to $1$ when $v$ is operational before propagation begins and $0$ when it has
already failed. We write $\mathbf{1}$ for the initial state that holds every node
operational, so $\mathbf{1}(v) = 1$ at every $v \in V$.
Propagation applies the same rule to every node in parallel and repeats it, round
$n+1$ reading the states of round $n$:
\begin{equation}
f_{n+1}(v) \;=\; f_n(v) \;\wedge \bigwedge_{s \,\in\, S_{\mathrm{req}}(v)} \;\;
          \bigvee_{u \,\in\, P_s(v)} f_n(u).
\label{eq:model}
\end{equation}
A node is operational in round $n+1$ when it was operational to begin with and
every service it requires continues to be available in round $n$ from at least one operational
supplier. The disjunction within a service models a parallel system: any one
operational supplier suffices, and the disjunction is false when every supplier in
$P_s(v)$ has failed, the empty $P_s(v)$ of a service with no supplier at all
included. The conjunction across services represents a series system: every service
a node requires must be met, and it is true over the empty $S_{\mathrm{req}}(v)$ of a
node with no incoming links, which therefore holds the state it started in. Every
node reads round $n$ to write round $n+1$, so $f_0$ determines the whole sequence.

The rule is monotone in the supplier states: switching a supplier off can only
switch its consumers off. Separately, its leading conjunct $f_n(v)$ bounds every
round from the previous state, so $f_{n+1} \le f_n$ pointwise and a node that fails
stays failed. A round that switches nothing off leaves every later round unchanged,
and any other round switches off at least one of the $|V|$ nodes, so the sequence
$f_0, f_1, f_2, \dots$ reaches a fixed point after at most $|V|$ rounds. Write
$f_\infty$ for the \emph{propagated state}, the fixed point it reaches, and $N(D)$ for the number of
nodes $f_\infty$ leaves operational.

\paragraph{What the model assumes.}
Four assumptions carry Eqs.~(\ref{eq:requires}) and~(\ref{eq:model}): each node
provides one service (easily generalizable to multiple services); 
the required services of each node are the ones labelling its suppliers; the suppliers
within $P_s(v)$ are interchangeable, so any one of them meets $s$ in full; and the
state is Boolean.
Section~\ref{sec:conclusion} returns to what these four leave out.

\begin{figure}[t]
\centering
\begin{tikzpicture}[
  nd/.style={circle,draw,minimum size=6.5mm,inner sep=0pt,font=\small},
  lk/.style={-{Latex[length=2mm]},thick},
  sv/.style={font=\scriptsize,fill=white,inner sep=1pt,pos=0.55}]

\node[nd] (a1) at (0,1.05) {$a$};
\node[nd] (a2) at (0,-0.15) {$b$};
\node[nd] (v)  at (2.4,0.45) {$v$};
\draw[lk] (a1) -- node[sv]{water} (v);
\draw[lk] (a2) -- node[sv]{water} (v);
\node[align=center,font=\footnotesize] at (1.2,-1.05)
  {$S_{\mathrm{req}}(v)=\{\textsf{water}\}$\\[-1pt] parallel: $v$ survives either loss};

\node[nd] (c1) at (6.0,1.05) {$p$};
\node[nd] (c2) at (6.0,-0.15) {$q$};
\node[nd] (w)  at (8.6,0.45) {$w$};
\draw[lk] (c1) -- node[sv]{power} (w);
\draw[lk] (c2) -- node[sv]{telemetry} (w);
\node[align=center,font=\footnotesize] at (7.3,-1.05)
  {$S_{\mathrm{req}}(w)=\{\textsf{power},\textsf{telemetry}\}$\\[-1pt] series: either loss stops $w$};

\draw[dashed,gray!60] (4.2,-1.55) -- (4.2,1.5);
\end{tikzpicture}
\caption{Two dependency structures sharing one graph topology. The services
carried by the links entering $v$ and $w$ make $v$ a parallel system and $w$ a
series system: $P_{\textsf{water}}(v)=\{a,b\}$ is one group of two alternatives,
while $P_{\textsf{power}}(w)=\{p\}$ and $P_{\textsf{telemetry}}(w)=\{q\}$ are two
groups that are jointly required.}
\label{fig:thesis}
\end{figure}
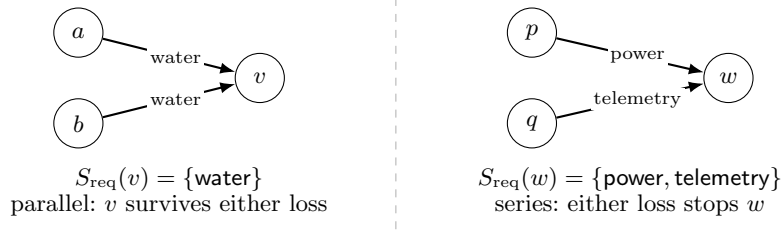

\section{Vitality on a Dependency Network}
\label{sec:vitality}

A \emph{vitality index} scores a node by how far some invariant function, chosen to
score a whole network, drops when that node is
removed~\cite{koschutzki_centrality_2005,skibski_vitality_2021}. The invariant is
the choice the analyst makes, and here it is $N(\cdot)$: the nodes
Eq.~(\ref{eq:model}) leaves operational from the fully operational initial state
$f_0 = \mathbf{1}$. Remove a node, and measure the operativity lost:
\begin{equation}
c_v(D) \;=\; \frac{N(D) - N(D \ominus v)}{|V|} \,,
\label{eq:index}
\end{equation}
where $D \ominus v$ deletes $v$ from $V$, its incident links from $E$, and its
label from $\sigma$, while $S_{\mathrm{req}}(\cdot)$ carries over unchanged: it records what
the surviving nodes require, and the absence of $v$ does not alter that. A consumer
$w$ drawing service $s$ from $v$ alone therefore keeps $s$ in $S_{\mathrm{req}}(w)$
while $P_s(w)$ loses its only member, which is the empty disjunction of
Section~\ref{sec:model} and fails $w$. Re-deriving the requirements from the
surviving links instead would leave every network fully operational and every node
scoring $1/|V|$.

The loss counts $v$ along with the nodes its absence brings down, so $c_v$ is the
share of the original nodes lost with $v$, $v$ included: a node whose failure
reaches no other scores the floor value $1/|V|$. In Section~\ref{sec:benchmarks} we will measure
the \emph{Mean Operativity Loss} $\Omega$, the average of $c_v$ over the network:
\begin{equation}
  \Omega(D) = \frac{\sum_{v \in V} c_v(D)}{|V|}.
\end{equation}

Eq.~(\ref{eq:index}) is written as a removal, while Eq.~(\ref{eq:model}) deletes
nothing: it propagates from an initial state. Write $N_{f_0(v)=0}(D)$ for the
operativity it reports on $D$ from the initial state with $f_0(v) = 0$ and
$f_0(u) = 1$ for every $u \neq v$. The two agree, so the vitality index of
Eq.~(\ref{eq:index}) is what the model reports when a single node is switched off.

\begin{proposition}
\label{prop:conditioning}
$N(D \ominus v) = N_{f_0(v)=0}(D)$ for every $v \in V$.
\end{proposition}

\begin{proof}
First, $v$ stays failed on the right: $f_n(v)$ is the leading conjunct of
Eq.~(\ref{eq:model}), so $f_0(v) = 0$ leaves $f_n(v) = 0$ in every round.

Next, take any other node $w$ and any service $s \in S_{\mathrm{req}}(w)$. If
$v \in P_s(w)$, then $v$ contributes $0$ to $\bigvee_{u \in P_s(w)} f_n(u)$, so that
disjunction takes the value of $\bigvee_{u \in P_s(w) \setminus \{v\}} f_n(u)$. The
smaller group is the one $w$ has in $D \ominus v$, because the deletion removes the
links of $v$ and no others; and $S_{\mathrm{req}}(w)$ is the same on both sides,
because $D \ominus v$ carries it. Each $w \neq v$ is therefore governed by the same
rule on both sides.

Both then run that rule from the same initial state and reach the same fixed point
on $V \setminus \{v\}$. Since $v$ is failed on one side and absent from the other,
it adds nothing to either count. \qed
\end{proof}

Proposition~\ref{prop:conditioning} is what lets the index be computed by setting
one Boolean state to $0$ and running
Eq.~(\ref{eq:model}), which is how every vitality reported in
Section~\ref{sec:benchmarks} is obtained. Two networks with the same nodes and the
same links can still give different values of $c_v(D)$, as Fig.~\ref{fig:thesis}
already shows: $c_v$ is a function of the labelled dependency graph.
Section~\ref{sec:benchmarks} measures how far apart the two can be driven.

\section{Benchmarks}
\label{sec:benchmarks}

We did not find a benchmark built to isolate the effect of service labels on a
dependency importance ranking, so we generate the networks ourselves. Each generator turns one knob and holds everything else fixed,
which is the standard way to isolate a single
mechanism~\cite{lancichinetti_benchmark_2008}, as widely adopted in the
literature~\cite{bachmann_survey_2020}.

Five indices computed from the graph topology serve as the comparison with the
vitality index $c_v$. \emph{Reachability} counts the nodes downstream of $v$ along
directed links, $v$ itself excluded: everything its failure could reach.
\emph{PageRank} is run at damping $0.85$ on the graph with every link reversed, so weight flows
from a consumer back to its suppliers and a node scores high when the nodes
depending on it score high. \emph{In-degree} counts the links entering $v$, that
is, how many suppliers it has, and \emph{out-degree} the links leaving it, how many nodes it supplies. 
\emph{Betweenness} is the share of directed shortest paths running through $v$ as an
interior node. All five are computed with NetworkX on the graph as given, and all
five are functions of that graph alone.

We will report the standard Kendall $\tau$ correlation between each of those five indices
and $c_v$, and the overlap between their top $10\%$ nodes. To handle ties, the correlation reported throughout is
Kendall's $\tau_b$, which carries the ties in its denominator:
\begin{equation}
\tau_b = \frac{n_c - n_d}{\sqrt{(n_0 - n_1)(n_0 - n_2)}},
\qquad n_0 = \tfrac{1}{2} n(n-1),
\end{equation}
where $n_c$ is the number of pairs of nodes whose order is the same in both
rankings, $n_d$ the number of pairs whose order is reversed, $n$ the number of
nodes, and $n_1$ and $n_2$ the numbers of pairs tied in the first and in the second
ranking.
The top-$10\%$ overlap is the share of the $k = \lfloor n/10 \rfloor$ highest-ranked
nodes that two rankings hold in common, $k = 6$ at the $n = 60$ used below. Where
nodes tie for the last places, each tied node enters with the probability a uniform
tie-break gives it, and the overlap reported is the expectation over those
tie-breaks, evaluated in closed form.

\subsection{Relabelling a Fixed Graph}
\label{sec:ablation}
This experiment holds one graph topology fixed and varies only the number of service labels.
Results are means over $20$ replicates; the bracketed intervals in the text are $95\%$ percentile
bootstrap intervals over those replicates.

In every replicate a link $u \to v$ makes $\sigma(u)$ a requirement of $v$ (Eq.~(\ref{eq:requires})),
so a node receiving a link from a supplier of service $X$ depends on $X$.
Each replicate draws one random acyclic graph on $n = 60$ nodes with some $109$
links and holds it fixed\footnote{The generator is in
\texttt{experiments/centrality/category\_ablation.py} of the repository named in
Section~\ref{sec:conclusion}.}. The only quantity that varies is $L$, the number
of services the labels are drawn from. At $L=1$ every supplier of a node provides
the same service; as $L$ grows the suppliers
increasingly provide distinct services and it becomes a series system. Every
index computed from the graph topology is therefore identical across $L$ changes. 

\begin{figure}[!t]
\centering
{\small
\begin{tabular}{@{}clcccccc@{}}
\toprule
$L$ & reading & $\Omega$ & $\tau$(reach) & $\tau$(PageRank) & $\tau$(out-deg) & $\tau$(betw.) & $\tau$(in-deg) \\
\midrule
1 & parallel    & 0.023 & 0.454 & 0.547 & 0.498 & 0.259 & $-0.134$ \\
2 &             & 0.049 & 0.693 & 0.736 & 0.717 & 0.404 & $-0.164$ \\
3 &             & 0.072 & 0.791 & 0.790 & 0.768 & 0.454 & $-0.156$ \\
5 &             & 0.097 & 0.845 & 0.833 & 0.811 & 0.472 & $-0.185$ \\
8 & near-series & 0.114 & 0.907 & 0.859 & 0.837 & 0.491 & $-0.189$ \\
\bottomrule
\end{tabular}}

\vspace{1ex}
\begin{tikzpicture}
\begin{axis}[
  width=0.80\textwidth, height=5.4cm,
  xlabel={number of services $L$}, ylabel={$\tau$ against $c_v$},
  xmin=0.5, xmax=8.5, ymin=-0.3, ymax=1.0,
  xtick={1,2,3,5,8}, ytick={-0.2,0.2,0.6,1.0},
  tick label style={font=\scriptsize}, label style={font=\footnotesize},
  grid=major, grid style={gray!20},
  legend style={font=\footnotesize, at={(0.99,0.34)}, anchor=east, draw=none,
                fill=white, fill opacity=0.85, text opacity=1}]
\addplot[thick,mark=*,mark size=1.5pt,black] coordinates
  {(1,0.454)(2,0.693)(3,0.791)(5,0.845)(8,0.907)};
\addlegendentry{reachability}
\addplot[thick,mark=square*,mark size=1.5pt,black,dashed] coordinates
  {(1,0.547)(2,0.736)(3,0.790)(5,0.833)(8,0.859)};
\addlegendentry{PageRank}
\addplot[thick,mark=diamond*,mark size=1.8pt,black,dashdotted] coordinates
  {(1,0.498)(2,0.717)(3,0.768)(5,0.811)(8,0.837)};
\addlegendentry{out-degree}
\addplot[thick,mark=o,mark size=1.5pt,black,densely dotted] coordinates
  {(1,0.259)(2,0.404)(3,0.454)(5,0.472)(8,0.491)};
\addlegendentry{betweenness}
\addplot[thick,mark=triangle*,mark size=2pt,black,dotted] coordinates
  {(1,-0.134)(2,-0.164)(3,-0.156)(5,-0.185)(8,-0.189)};
\addlegendentry{in-degree}
\end{axis}
\end{tikzpicture}
\caption{The graph held fixed, the service labels varying. $\Omega$
is the average of $c_v$ over all nodes and $\tau$ Kendall's rank correlation
between $c_v$ and the reported index computed from the graph topology. Within each
replicate the graph is identical at every $L$, so each of those indices is
constant throughout; the vitality ranking they are compared against is what moves.
The $\Omega$ column rises fivefold as the labels turn a parallel
network into a series one. Each cell is a mean over $20$ replicates.}
\label{fig:ablation}
\end{figure}
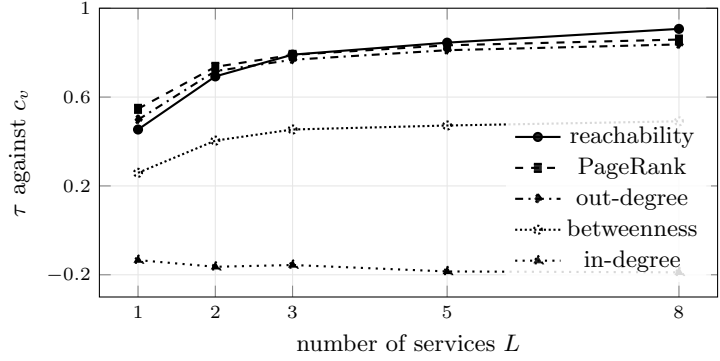

Figure~\ref{fig:ablation} reports the sweep. On an unchanged graph, relabelling
moves $\Omega$ from $0.023$ $[0.023, 0.024]$ to $0.114$ $[0.106, 0.124]$, a factor
of five, while every index computed from the graph topology holds the value it had
at $L=1$. Comparing the extremes directly, the rankings at $L=1$ and $L=8$ agree at
$\tau = 0.476$ $[0.447, 0.505]$ and their top-$10\%$ sets share $0.402$ $[0.346,
0.457]$ of their members, so relabelling moves three fifths of the six nodes
selected for protection. 

Running the same generator with up to four suppliers per
node in place of three gives $\tau = 0.399$ $[0.357, 0.440]$, an overlap of $0.301$
$[0.237, 0.363]$, and a loss rising from $0.021$ to $0.138$: a graph with more
incoming links per node increases the labels' influence.

\subsection{The Series Limit}
\label{sec:series}

Where Section~\ref{sec:ablation} shows that the labels move the ranking, this
experiment asks when a topology index still recovers it. It varies how much
redundancy the network holds and measures how much of the vitality $c_v$ ranking
survives in common topology indices. The knob is $\rho$, the share of required
services that have exactly one supplier: the generator gives each required service
of each node one supplier with probability $\rho$ and two with probability $1 -
\rho$, drawn independently, so $\rho$ is the expected share across the network. At
$\rho = 1$ every required service has a sole supplier, so losing any supplier
fails its consumers, reproducing the reading of Parshani et al.~\cite{parshani_critical_2011}. As
$\rho$ falls, required services acquire a second supplier and parallel structure
enters. At $\rho = 0$ every required service has two suppliers, single losses have no consequences.

A second knob, $p$, sets how many services a node requires: with probability $p$
it requires two and with probability $1 - p$ one, so $p$ sets the expected share of
nodes requiring two services while $\rho$ fixes how many suppliers each required
service has.\footnote{The generator is in
\texttt{experiments/centrality/regime\_probe.py} of the repository named in
Section~\ref{sec:conclusion}, run with \texttt{-{}-reps 20}}.
Each row of Table~\ref{tab:series} fixes $\rho$ and pools $20$ replicates at $p = 0$ with $20$ at $p = 0.5$, $40$ in all, so a
row reports one level of redundancy across both settings of how many services a
node requires. 
\begin{table}[!b]
\centering
\caption{Rank agreement between $c_v$ and five indices computed from the graph
topology, against the redundancy level $\rho$ defined above. $\tau$ is Kendall's
$\tau_b$. Each cell is a mean over the $40$ replicates pooled into its row; a $95\%$
percentile bootstrap interval over them reaches $\pm 0.03$, and $\pm 0.07$ in the
in-degree and betweenness columns.}
\label{tab:series}
{\small
\begin{tabular}{@{}lccccc@{}}
\toprule
$\rho$ & $\tau$(reach) & $\tau$(PageRank) & $\tau$(out-deg) & $\tau$(betw.) & $\tau$(in-deg) \\
\midrule
0.0 & 0.195 & 0.204 & 0.203 & $-0.048$ & $-0.212$ \\
0.2 & 0.378 & 0.430 & 0.403 & 0.163 & $-0.199$ \\
0.4 & 0.533 & 0.602 & 0.563 & 0.272 & $-0.222$ \\
0.6 & 0.670 & 0.725 & 0.682 & 0.422 & $-0.177$ \\
0.8 & 0.811 & 0.845 & 0.782 & 0.516 & $-0.224$ \\
1.0 & 1.000 & 0.960 & 0.900 & 0.626 & $-0.288$ \\
\bottomrule
\end{tabular}}
\end{table}

\emph{Reachability is exactly correct at the series limit.} The model sends failure
along a link only when the target has no other supplier of that service, and at
$\rho = 1$ no target has one, so every link into an operational node carries failure
whatever the rest of the network does. Removing $v$ therefore fails $v$ together
with every node reachable from it. Writing $R(v)$ for the nodes downstream of $v$
with $v$ itself excluded, which is the reachability index as defined above,
\begin{equation}
N(D) - N(D \ominus v) \;=\; 1 + |R(v)| ,
\label{eq:series}
\end{equation}
so $c_v = (1 + |R(v)|)/|V|$ is an affine function of $|R(v)|$ with the same two
constants at every node. The two therefore order the nodes identically: $\tau =
1.000$. The other four indices stop short, PageRank at
$0.960$ and out-degree at $0.900$. Eq.~(\ref{eq:series}) is checked node by node in
every replicate at $\rho = 1$, where it holds exactly. The limit is where a
deterministic spreading process is modelled: contamination entering a water network
reaches everything downstream of the point it enters.

\emph{Its accuracy falls as redundancy enters}, reaching $\tau = 0.195$ at
$\rho = 0$, and the decline is monotone across the six levels of $\rho$ swept here.
Each service that gains a second supplier makes one of the links carrying it inert
against a single failure, and a ranking computed from the graph topology counts that
link anyway. The error such a ranking makes is therefore a function of how the links
entering a node group into services: the graph gives their number and their
arrangement, and the grouping decides which of them a single failure travels along.

\section{Applicability, Limitations and Conclusion}
\label{sec:conclusion}

This paper follows one observation to its consequence: two dependency links
entering a node are alternatives when they carry the same service and are
jointly required when they carry different services. Section~\ref{sec:model}
turns that observation into requirements read off the links,
Section~\ref{sec:vitality} into a vitality index on the operativity, and
Section~\ref{sec:benchmarks} into a number: on an unchanged graph topology the
labels alone move $\Omega$ by a factor of five. 
This section sets out where the construction applies, what it assumes, and where it stops.

\paragraph{Where it applies.}
The network construction needs three ingredients: a set of services, directed
dependency links, and a notion of failure. In interdependent essential infrastructure the
services are the commodities delivered: a hospital requires electricity and
water, two required services, and it draws that electricity from either the grid
connection or a standby generator, two alternatives within one service. In a
software deployment the services are the components a request passes through: a
deployment answers when its database and its authentication service both answer,
and either replica of each serves. In each case the analyst supplies the label a
node carries, and Eq.~(\ref{eq:requires}) reads the requirements off the links.

What the analyst accepts along with the labels is the propagation rule. The
disjunction and conjunction of Eq.~(\ref{eq:model}) is the assumption made about
how failure combines at a node, and it is written where an analyst can accept it
or reject it, which is Borgatti's process
matching~\cite{borgatti_centrality_2005} made into a step of the work. An index
computed from the unlabelled graph topology makes its own assumption about how
failure moves, and makes it silently.

\paragraph{Where it stops.}
Eq.~(\ref{eq:model}) checks that each required service has an operational supplier,
and a supplier is operational when its own requirements are met. Support is local
in that sense, and it can close on itself: a group of nodes supplying one another
satisfies the rule from the inside and holds itself up after the source feeding the
group is gone. A loop of interconnected junctions behind a failed reservoir stays
operational here, which is wrong. Tracing a service back to where it originates is
what that case needs, and a flow-based propagation carrying quantity along the
links supplies it.

One property runs the other way: Eq.~(\ref{eq:model}) accepts any initial state
$f_0$, so a scenario with several nodes failed together can be analysed directly,
and the vitality index of Section~\ref{sec:vitality} is its single-node case.

In conclusion, the description of a network of networks should carry the service
each node provides alongside its dependency links. The links give a node its
suppliers; the labels group those suppliers into the services the node requires,
and that grouping decides, for every link entering the node, whether the link is a
series element or a parallel one, and through that the vitality of the node: on one
fixed graph topology, relabelling the services moves three fifths of the top-$10\%$
nodes. For system reliability, the labels decide which links are single points of failure:
a link is one when its supplier is the only provider of its service to the node it enters.
The same graph therefore tolerates more or fewer single failures depending on the labels,
and $\Omega$, the mean share of the system lost to one failure, rises from $0.023$ to
$0.114$. Ranking the nodes worth protecting begins with defining what each node provides.

\paragraph{Data and code availability.}
The generators, the propagation model, and the scripts reproducing every number
reported in Section~\ref{sec:benchmarks} are available at
\url{https://github.com/cascade-platform-org/CASCADE}.

\bibliographystyle{spmpsci}
\bibliography{refs}

\end{document}